\documentclass[pra,superscriptaddress,showpacs,twocolumn]{revtex4}

\usepackage{amsfonts,amssymb,amsmath}
\usepackage{graphics,graphicx,epsfig}
\usepackage{amsthm}
\usepackage{color}

\def\identity{\leavevmode\hbox{\small1\kern-3.8pt\normalsize1}}

\newtheorem{theorem}{Theorem}

\newtheorem{conjecture}{Conjecture}

\newcommand{\ket}[1]{\left | #1 \right\rangle}
\newcommand{\bra}[1]{\left \langle #1 \right |}

\newcommand{\Tr}{\mathrm{Tr}}
\newcommand{\braket}[2]{\left\langle #1|#2\right\rangle}
\newcommand{\proj}[1]{\ket{#1}\bra{#1}}
\renewcommand{\epsilon}{\varepsilon}

\begin{document}

\title{Genuine $N$-partite entanglement without $N$-partite correlation functions}

\author{Minh~Cong~Tran}
\affiliation{School of Physical and Mathematical Sciences, Nanyang Technological University, 637371 Singapore}
\affiliation{Joint Center for Quantum Information and Computer Science, University of Maryland, College Park, Maryland 20742, USA}

\author{Margherita~Zuppardo}
\affiliation{School of Physical and Mathematical Sciences, Nanyang Technological University, 637371 Singapore}
\affiliation{Science Institute, University of Iceland, Dunhaga 3, IS-107 Reykjavik, Iceland}

\author{Anna~de~Rosier}
\affiliation{Institute of Theoretical Physics and Astrophysics, Faculty of Mathematics, Physics and Informatics, University of Gda\'nsk, 80-308 Gda\'nsk, Poland}

\author{Lukas~Knips}
\affiliation{Max-Planck-Institut f\"ur Quantenoptik, Hans-Kopfermann-Stra{\ss}e 1, 85748 Garching, Germany}
\affiliation{Department f\"ur Physik, Ludwig-Maximilians-Universit\"at, 80797 M\"unchen, Germany}

\author{Wies{\l}aw~Laskowski}
\affiliation{Institute of Theoretical Physics and Astrophysics, Faculty of Mathematics, Physics and Informatics, University of Gda\'nsk, 80-308 Gda\'nsk, Poland}

\author{Tomasz~Paterek}
\affiliation{School of Physical and Mathematical Sciences, Nanyang Technological University, 637371 Singapore}
\affiliation{MajuLab, CNRS-UNS-NUS-NTU International Joint Research Unit, UMI 3654 Singapore, Singapore}

\author{Harald~Weinfurter}
\affiliation{Max-Planck-Institut f\"ur Quantenoptik, Hans-Kopfermann-Stra{\ss}e 1, 85748 Garching, Germany}
\affiliation{Department f\"ur Physik, Ludwig-Maximilians-Universit\"at, 80797 M\"unchen, Germany}

\begin{abstract}
A genuinely $N$-partite entangled state may display vanishing $N$-partite correlations measured for arbitrary local observables. In
such states the genuine entanglement is noticeable solely in correlations between subsets of particles. A straightforward way to
obtain such states for odd $N$ is to design an `anti-state' in which all correlations between an odd number of observers are exactly
opposite. Evenly mixing a state with its anti-state then produces a mixed state with no $N$-partite correlations, with many of them
genuinely multiparty entangled. Intriguingly, all known examples of `entanglement without correlations' involve an \emph{odd} number of
particles. Here we further develop the idea of anti-states, thereby shedding light on the different properties of even and odd particle
systems. We conjecture that there is no anti-state to any pure even-$N$-party entangled state making the simple construction scheme
unfeasable. However, as we prove by construction, higher-rank examples of `entanglement without correlations' for arbitrary even $N$
indeed exist. These classes of states exhibit genuine entanglement and even violate an $N$-partite Bell inequality, clearly demonstrating
the non-classical features of these states as well as showing their applicability for quantum communication complexity tasks.
\end{abstract}

\pacs{03.65.Ud}

\maketitle


Quantum entanglement is present in quantum states that cannot be obtained from uncorrelated states by local operations and classical communication~\cite{locc,RevModPhys.81.865}.
It turns out that for pure states the existence of entanglement is fully captured by $N$-partite correlation functions only:
A pure state is entangled if and only if the sum of squared $N$-partite correlation functions exceeds certain bound~\cite{QuantInfComp.8.773,PhysRevA.77.062334,I,PhysRevA.92.050301,Random2016}.
One may then wonder if similar detection methods could exist for mixed states, i.e. whether appropriate processing of only $N$-partite correlation functions detects entanglement in all mixed states.
The states we consider here demonstrate vividly that such a universal entanglement criterion does not exist.
Despite vanishing $N$-partite correlation functions in all possible local measurements, these states can be even genuinely $N$-partite entangled.
As a matter of fact the genuine $N$-partite entanglement is due to non-vanishing correlations between less than $N$ particles, so-called lower order correlations.

The first example of such a state was given in~\cite{PhysRevLett.101.070502} and consists of an even mixture of two $W$ states between an odd number of qubits.
The two states have exactly opposite $N$-partite correlations such that they average out in the even mixture.
More recently it was shown that any pure quantum state has an `anti-state' where all correlation functions have opposite signs, but only between an \emph{odd} number of observers~\cite{PhysRevA.86.032105,Us2015}.
Then, the equal mixture of a pure state with odd number of qubits and its anti-state produces a mixed state with vanishing $N$-partite correlation functions.
Many of such `no-correlation' states are genuinely $N$-partite entangled and even an infinite family of such states with two continuous parameters could be constructed~\cite{Us2015}.

Here we generalise the notion of anti-states and study their relations to entanglement without correlations.
A number of problems was raised in Refs.~\cite{PhysRevLett.101.070502,PhysRevA.86.032105,Us2015} which have now been solved.
In particular, we provide analytical plausibility argument (Sec. \ref{general}) and strong numerical evidence (Conjecture \ref{CON_CON}) that there is no anti-state to any genuinely multiparty entangled pure state of an even number of qubits.
This explains why previous relatively simple examples of entanglement without correlations could be constructed for odd number of qubits only.
Therefore, in the case of an even number of particles, this phenomenon requires mixing of at least three pure quantum states.
We provide here analytical examples of rank-$4$ mixed states that are both genuinely $N$-partite entangled and have vanishing all $N$-partite correlation functions, for arbitrary even $N$ (Sec. \ref{ewc}).
Up to numerical precision also rank-$3$ mixed states with this property exist.
Remarkably, they violate suitably designed Bell-type inequality (Sec. \ref{SEC_BELL}).
In order to further emphasise that entanglement without correlations is not very unusual, we extend the previous example of the infinite family with two continuous parameters to exponentially many in $N$ continuous parameters (Sec. \ref{odd}).
This is achieved with the help of a simple identifier of genuine multipartite entanglement that also illustrates limits to entanglement detection with only bipartite correlation functions~\cite{PhysRevA.72.022340,JOptSocAmB.24.275,PhysRevA.86.042339,PhysRevA.87.034301,Science.344.1256,JPhysA.47.424024}.


\section{Basic notions}

\subsection{Genuine multipartite entanglement}

A mixed quantum state of $N$ particles is genuinely $N$-partite entangled if it cannot be written as:
\begin{equation}
\rho \ne \sum_j p_j \rho_{A_j}^j \otimes \rho_{B_j}^j,
\end{equation}
where $A_j : B_j$ is a partition of the $N$ particles and $p_j$'s are probabilities.
Note that different terms in this convex decomposition may involve different partitions.
The states $\rho_{A_j}^j \otimes \rho_{B_j}^j$ can always be chosen pure, in which case they are called bi-product states.
All our examples will exploit the fact that if the support of $\rho$ does not contain a single bi-product state, then $\rho$ must be genuinely $N$-partite entangled.

\subsection{Correlation functions}

The correlation function is a standard statistical quantifier defined as the expectation value of a product of measurement results.
Consider dichotomic observables, i.e. the measurement results are $\pm 1$, conducted on multiple qubits.
Such observables are parameterised by unit vectors on a sphere.
We denote by $\vec m_n$ the vector encoding the observable of the $n$th party.
If such observables are measured on every particle from a $N$-partite quantum system in state $\rho$ one obtains the $N$-partite (quantum) correlation function:
\begin{equation}
E(\vec m_1, \dots, \vec m_N) = \Tr(\rho \, \vec m_1 \cdot \vec \sigma \otimes \dots \otimes \vec m_N \cdot \vec \sigma),
\end{equation}
where $\vec \sigma=\left(\sigma_x,\sigma_y,\sigma_z\right)$ is the vector of Pauli operators.
We shall also write $\sigma_x, \sigma_y, \sigma_z$ as $\sigma_1, \sigma_2, \sigma_3$, respectively.
It is customary to introduce correlation tensor $T$ or, respectively, its coefficients
\begin{equation}
T_{j_1 \dots j_N}(\rho) = \Tr(\rho \, \sigma_{j_1} \otimes \dots \otimes \sigma_{j_N}),
\end{equation}
for the $N$-partite correlation functions measured explicitly along the $\vec x,\vec y,\vec z$ axes. Here $j_n = 1,2,3$.
By writing $\rho$ in the basis of tensor products of Pauli operators, one easily verifies the tensor transformation law:
\begin{equation}
E(\vec m_1, \dots, \vec m_N) = \sum_{j_1, \dots, j_N = 1}^3 T_{j_1 \dots j_N} (\vec m_1)_{j_1} \dots (\vec m_N)_{j_N},
\end{equation}
where $(\vec m_n)_{j_n}$ is the component of the vector $\vec m_n$ along the $j_n$th axis.
In the present context this implies that it is sufficient to ensure that $T_{j_1 \dots j_N} = 0 $ for all $j_1,\dots,j_N = 1,2,3$ to guarantee that $N$-partite correlation functions vanish for arbitrary local measurements.

One could of course also measure subsets of all $N$ particles, in which case the resulting correlation functions are called lower order correlations.
We will be only interested in these correlations along the $\vec x,\vec y,\vec z$ axes, in which case they can be calculated as follows:
\begin{equation}
T_{\mu_1 \dots \mu_N}(\rho) = \Tr(\rho \, \sigma_{\mu_1} \otimes \dots \otimes \sigma_{\mu_N}),
\end{equation}
where index $\mu_n = 0,1,2,3$, i.e. additionally to Pauli operators it also includes $\sigma_0$, the identity operator, for those parties who do not conduct measurements.
For example, bipartite correlation functions between the first two observers are denoted by $T_{j_1 j_2 0 \dots 0}$.

\subsection{Anti-states}

Given a pure or mixed state $\rho$, with the $N$-partite correlation functions $T_{j_1 \dots j_N}$,
we define its anti-state $\bar \rho$ by the requirement that all its $N$-partite correlation functions have opposite sign, i.e. $T_{j_1 \dots j_N} (\bar\rho)= - T_{j_1 \dots j_N}(\rho)$ for all indices $j_n = 1,2,3$.
No assumptions are made about the lower order correlation functions.

Ref. \cite{Us2015} presented a method to build an anti-state to an arbitrary pure state with an odd number $N$ of qubits. 
However, this method does not apply to cases where $N$ is even. 
We therefore need to use different approaches depending on the parity of $N$, as we will discuss in the next sections.


\section{N even}

Let us begin with the problem of existence of anti-states for an even number of qubits.
We argue that most likely all genuinely $N$-partite entangled pure states of even $N$ do not admit anti-states.
Nevertheless, this does not imply impossibility of entanglement without correlations.
It just says that more than two pure states have to be present in the mixture.
Indeed, we will provide such examples for every even $N \ge 4$.

We start with bipartite systems where one can easily exclude existence of an anti-state to arbitrary pure entangled state.
\begin{theorem}
There is no anti-state to an arbitrary, entangled pure state of two qubits.
\end{theorem}
\begin{proof}
Any pure state can be written in the Schmidt form $\ket{\psi} = a\ket{00}+b\ket{11}$, with real coefficients.
In this basis, the only non-zero elements of the correlation tensor are 
\begin{eqnarray}
\label{2q}
T_{zz}(\psi) & = & 1, \\
T_{xx}(\psi) & = & -T_{yy}(\psi) = 2 a b.
\end{eqnarray}
Therefore, the hypothetical anti-state (mixed states allowed) has to have $T_{zz} = -1$, and hence it lies in the subspace spanned by $\ket{01}$ and $\ket{10}$.
Since all such states have $T_{xx} = T_{yy}$, only the product state with $ab=0$ has an anti-state. %
\end{proof}

There is strong numerical evidence that a pure genuinely $N$-partite entangled state of $N=4$ and $N=6$ qubits does not admit an anti-state.
We are therefore conjecturing this in general. 
\begin{conjecture}
\label{CON_CON}
There is no anti-state to an arbitrary genuinely $N$-partite entangled pure state of even-$N$ qubits.
\end{conjecture}
\begin{proof}[Evidence] 
Our aim is to verify to a high precision whether an anti-state exists to a preselected state $\ket{\psi}$.
In our numerical approach we parameterize a candidate state $\ket{\phi}$ and use Simulated Annealing \cite{Kirkpatrick} to globally minimize the length of correlation \cite{PhysRevA.92.050301} 
of the even mixture $\rho = \frac{1}{2}\ket{\psi}\bra{\psi}+\frac12 \ket{\phi}\bra{\phi}$.
The length of correlation is defined as:
\begin{equation}
L(\rho) = \sum_{j_1,\dots,j_N=1}^3 T_{j_1\dots j_N}^2(\rho).
\end{equation}
If anti-states to $\ket{\psi}$ exist, then $L(\rho)$ will converge to 0 while $\ket{\phi}$ converges to an anti-state. 

We tested this algorithm on states of $N=3$ and $N=5$ qubits, using the genuinely $N$-partite entangled input state $\ket{\psi}$.
The candidate state $\ket{\phi}$ converged to an approximate anti-state, in accordance with what is known about anti-states with an odd number of qubits.
For $N=4$ and $N=6$ we tested both the states for which it is known that they have no anti-state, the Greenberger-Horne-Zeilinger states~\cite{Us2015}, the W states, the Dicke states as well as a thousand of randomly chosen states.
We found that their length of correlation always converged to a finite value larger than $0$, which indicates that there is no anti-state.

On the other hand, in addition to $\ket{\phi}$ we also varied the state $\ket{\psi}$ as we minimized the length of correlation.
In this case, for all choices of initial states, the algorithm quickly converged to a pair of state--anti-state.
However, all such states were not genuinely $N$-partite entangled.
Instead, each pair was of the form $\{\ket{\psi_{N-1}} \otimes \ket{\phi_1}, \ket{\overline{\psi}_{N-1}} \otimes \ket{\phi_1} \}$,
i.e. a bi-product of a state or anti-state between (odd) $N-1$ qubits and a common single-qubit state. 
This strongly suggests that a genuinely multi-partite entangled state of $N=4$ and $N=6$ qubits does not have an anti-state.
\end{proof}

\subsection{Entanglement without correlations}
\label{ewc}

Although anti-states to \textit{pure} $N$-partite entangled states most likely do not exist, one can find anti-states to \textit{mixed} entangled states. 
These can subsequently be used to construct examples of states with no $N$-partite correlation functions yet containing genuine $N$-partite entanglement.
The simplest such example involves $4$ qubits.
For two qubits, while anti-states to mixed entangled states can easily be constructed, all states with vanishing correlation functions are separable.
A simple anti-state example can be seen as follows.
Consider the state $\rho = \frac{1}{2} \proj{\psi^+} + \frac{1}{2} \proj{11}$.
Being a mixture of a pure entangled state and a product state, $\rho$ is entangled~\cite{TheorComputSci.292.589,JPhysA.47.424025}.
Its anti-state is given by $\bar \rho = \frac{1}{2} \proj{\psi^-} + \frac{1}{2} \proj{11}$ as can be directly verified.
The anti-state is also entangled by the same argument, but the even mixture of the two states, $\frac{1}{2}\left(\rho+\bar \rho\right)$, is separable.

The following theorem proves in general the absence of entanglement in bipartite states without bipartite correlation functions.
\begin{theorem}
Two-qubit states with vanishing bipartite correlation functions are separable.
\end{theorem}
\begin{proof}
The most general bipartite state with vanishing bipartite correlation functions is of the form:
\begin{equation}
\rho = \tfrac{1}{4} (\openone + \vec a \cdot \vec \sigma \otimes \sigma_0  + \sigma_0 \otimes \vec b \cdot \vec \sigma),
\end{equation}
where $\openone$ denotes identity operator in the space of two qubits, $|\vec a| \le 1$, and similarly $|\vec b| \le 1$.
It has eigenvalues $\frac{1}{4}(1 \pm ||\vec a|^2 \pm |\vec b|^2|)$ with all four sign combinations allowed.
The same eigenvalues are obtained after partially transposing $\rho$.
Hence all of such states are separable~\cite{PhysRevLett.77.1413,PhysLettA.223.1}.
\end{proof}

Theorem 2 does not generalize to $N>2$.
A similar code to the one used in the evidence for Conjecture~\ref{CON_CON} returned a rank-$3$ genuinely $4$-party entangled state of $4$ qubits with no $4$-partite correlation functions. 
Here we provide an analytical example of rank-$4$ state for arbitrary even $N \ge 4$, giving rise to entanglement without correlations.
Consider a mixed state 
\begin{equation}
\rho_0 = \frac{1}{4} \proj{\psi_1} + \dots + \frac{1}{4} \proj{\psi_4},
\label{EVEN_NO_CORR}
\end{equation}
which mixes the following pure states:
\begin{eqnarray}
	\ket{\psi_1} & = & \tfrac{1}{\sqrt{2}} \left( \ket{0 \dots 0} \ket{\psi} + \ket{\psi} \ket{0 \dots 0} \right), \nonumber \\
	\ket{\psi_2} & = & \tfrac{1}{\sqrt{2}} \left( \ket{1 \dots 1} \ket{\bar \psi} - \ket{\bar \psi} \ket{1 \dots 1} \right), \nonumber \\
	\ket{\psi_3} & = & |\phi \rangle | \psi \rangle,\nonumber \\
	\ket{\psi_4} & = & | \psi \rangle |\phi \rangle, \label{EVEN_STATE}
\end{eqnarray}
where we take the $\ket{\psi}$ state and its anti-state $\ket{\bar \psi}$ as generalised $W$ states of $N/2$ qubits
\begin{eqnarray*}
\ket{\psi} & = & \alpha_1 \ket{10\dots0} + \alpha_2 \ket{01\dots 0} + \dots + \alpha_{N/2} \ket{00 \dots 1}, \\
\ket{\bar \psi} & = & \alpha_1 \ket{01\dots1} + \alpha_2 \ket{10\dots 1} + \dots + \alpha_{N/2} \ket{11 \dots 0}.\label{EQ_psi1234}
\end{eqnarray*}
It is assumed that all the coefficients are real and strictly positive, i.e. $\alpha_n > 0$.
The state $\ket{\phi}$ is any product state containing an odd number of excitations, i.e. ones.
To show that $\rho_0$ is genuinely multipartite entangled, one may attempt to seek suitable entanglement witnesses. 
Here we present a much simpler approach. 
For that, we need the following theorem. 

\begin{theorem}
If a state $\rho$ lies in the subspace spanned by $\left\{\ket{\psi_1},\ket{\psi_2},\ket{\psi_3},\ket{\psi_4}\right\}$ given in Eq. \eqref{EVEN_STATE} and $\rho$ is bi-separable, then $\rho$ is orthogonal to $\ket{\psi_1}$, i.e.
\begin{align}
	\Tr\left(\rho\ket{\psi_1}\bra{\psi_1}\right)=0.
\end{align}
\label{SUBSPACE_TH}
\end{theorem}
\begin{proof}
We first prove that all bi-product pure states in this subspace are orthogonal to $\ket{\psi_1}$.
Suppose there exists a bi-product state $\ket{\xi}_A \ket{\eta}_B \in \mathrm{span}\left\{\ket{\psi_1},\dots,\ket{\psi_4}\right\}$, i.e. 
\begin{equation}
\label{contr}
\ket{\xi}_A \ket{\eta}_B = c_1 \ket{\psi_1} + c_2 \ket{\psi_2} + c_3 \ket{\psi_3} + c_4 \ket{\psi_4}.
\end{equation} 
Here $A, B$ form an \emph{arbitrary} biparition of all $N$ qubits.
We emphasise that say $A$ contains any subset of qubits, not even neighbouring ones.
Denote $a_0 = \braket{00\dots0}{\xi}_A$ and $b_0=\braket{00\dots0}{\eta}_B$.
Since the discussed subspace is orthogonal to the $\ket{0\dots0}$ state of all $N$ qubits, by taking the inner product with both sides of Eq. \eqref{contr}, we conclude that
\begin{equation}
a_0 b_0 = 0.
\label{A0B0}
\end{equation}
Now consider the vector $\bra{00\dots0}_A\bra{10\dots0}_B$, where the excitation $1$ is in the first qubit of the subsystem $B$.
Since this vector has only one excitation, it is orthogonal to $\ket{\psi_2}$ (which has $N-1$ excitations) and both $\ket{\psi_3}$ and $\ket{\psi_4}$ (they have an even number of excitations).
The inner product with both sides of Eq. \eqref{contr} gives
\begin{equation}
a_0 b_1 = c_1 \, \alpha_k, 
\label{A0B1}
\end{equation}
where $b_1=\braket{10\dots0}{\eta}_B$ and index $k$ depends on which bipartition is chosen.
For example, if $A$ contains first half of the qubits, then comparison with (\ref{EVEN_STATE}) shows that $k=1$, or if $A$ contains first $N-1$ qubits, then the same analysis reveals that $k=N$.
Since we assume that all $\alpha_n > 0$, the state $\ket{\psi_1}$ is a superposition of one excitation on every qubit and hence for arbitrary bipartition there exists index $k$ such that (\ref{A0B1}) holds.
Similarly, by taking the inner product with the vector $\bra{10\dots0}_A\bra{00\dots0}_B$, we obtain 
\begin{equation}
a_1 b_0 = c_1 \, \alpha_l,
\label{A1B0}
\end{equation}
where  $a_1=\braket{10\dots0}{\xi}_A$ and $\alpha_l$ is the suitable coefficient of $\ket{\psi_1}$.
Multiplying (\ref{A0B1}) by (\ref{A1B0}) shows that $a_0 b_0 a_1 b_1=c_1^2 \alpha_k\alpha_l$. 
The left hand side of this equation is zero, as we showed in \eqref{A0B0}. 
Since both $\alpha_k$ and $\alpha_l$ are strictly positive numbers, we conclude that $c_1=0$.
In words, all bi-product states in the discussed subspace are orthogonal to $\ket{\psi_1}$.
Hence, arbitrary mixture of such states is also orthogonal to $\ket{\psi_1}$ and the theorem follows.
\end{proof}
According to this theorem, if a general state $\rho$ has non-zero overlap with $\ket{\psi_1}$, then either $\rho$ is genuinely $N$-qubit entangled or it does \emph{not} belong to the subspace spanned by $\{ \ket{\psi_1},\dots,\ket{\psi_4} \}$ or both.
Since the state $\rho_0$ presented in Eq.~\eqref{EVEN_NO_CORR} clearly belongs to this subspace, it has to be genuinely $N$-qubit entangled.
Furthermore, we prove in the Appendix that this state has no $N$-partite correlation functions.
This concludes construction of `entanglement without correlations' for any even $N \ge 4$.

The construction just given also sheds light on the kind of operations required to produce an anti-state.
In particular, one could consider a mixed state $\rho = \frac{1}{2} \proj{\psi_3} + \frac{1}{2} \proj{\psi_4}$, which is clearly bi-separable.
By our construction, its anti-state is $\bar \rho = \frac{1}{2} \proj{\psi_1} + \frac{1}{2} \proj{\psi_2}$, which is genuinely $N$-partite entangled.
Hence, at least some of the anti-states cannot be obtained by local operations and classical communication as this class of maps is not capable of producing entanglement.

\subsection{Violation of local realism}
\label{SEC_BELL}

Another remarkable property of states in Eq.~(\ref{EVEN_NO_CORR}) is their ability to violate a Bell inequality.
The lack of $N$-partite correlation functions makes many standard tools inapplicable to these states.
This was first pointed out in~\cite{PhysRevLett.101.070502} and only recently suitable Bell inequalities were found~\cite{PhysRevA.86.032105,Us2015}
and were experimentally implemented to test the no-correlation states of odd number of qubits~\cite{Us2015}.
We present now a Bell-type inequality which is violated by appropriate quantum measurements on states (\ref{EVEN_NO_CORR}) for arbitrary even $N$.

Consider the following Bell-type inequality introduced in~\cite{JPhysA.48.465301}:
\begin{equation}
0 \leq P(+ \dots + | A_1 \dots A_{N-2}) \textrm{CH}_{N-1,N}^{+\dots+}, 
\label{Wineq}
\end{equation}
where $P(+ \dots + | A_1 \dots A_{N-2})$ is the probability that the first $N-2$ parties all detect $+1$ outcomes when they measure observables $A_1, \dots, A_{N-2}$, respectively;
$\textrm{CH}_{N-1,N}^{+\dots+}$ denotes the Clauser-Horne expression~\cite{CH} between the last two parties, 
which is calculated in the subensemble of experiments in which the first $N-2$ observers all obtain $+1$.

For simplicity let us choose $|\psi \rangle$ as the symmetric W state of $N/2$ qubits, i.e. all $\alpha_n = 1/\sqrt{N/2}$.
In order to demonstrate a violation of (\ref{Wineq}), each of the first $N-2$ observers performs measurement $A_n = \sigma_z$.
Therefore, $P(+ \dots + | A_1 \dots A_{N-2}) = \frac{1}{4} \frac{2}{N}$ with the sole contribution from the state $\ket{\psi_1}$.
In the subensemble where all these $N-2$ results are $+1$ the state of the last two qubits collapses to $\frac{1}{\sqrt{2}}(|01\rangle + |10\rangle)$.
The last two observers perform measurements that lead to the maximal violation of the CH inequality given by $-\frac{\sqrt{2}-1}{2}$. 
Finally, the right-hand-side of (\ref{Wineq}) is equal to $- \frac{\sqrt{2}-1}{4 N}$, which violates the lower bound $0$.
We also verified, using the software described in~\cite{PhysRevA.82.012118}, that the above inequality is optimal
is the sense that it is violated for the highest admixture of white noise to the state $\rho$.

We note that additionally to fundamental interest this also demonstrates practical applicability of states (\ref{EVEN_NO_CORR}).
It is well-known that such states reduce communication complexity, improve security of cryptographic key distribution or enable device-independent protocols~\cite{RevModPhys.86.419}.


\section{N odd}
\label{odd}

Ref.~\cite{Us2015} demonstrated a continuous family of mixed states which are genuinely tripartite entangled and give rise to vanishing tripartite correlation functions.
In this section we will extend this example to a larger family of states described by exponentially many, in $N$, parameters.
This example will then be shown to elucidate limits on entanglement detection with bipartite correlation functions only, such as those discussed in~\cite{PhysRevA.72.022340,JOptSocAmB.24.275,PhysRevA.86.042339,PhysRevA.87.034301,Science.344.1256,JPhysA.47.424024}.

Consider the family of generalised Dicke states of $N$ qubits:
\begin{align}
	\ket{D^e_N} = \sum_{\mathcal{P}}\alpha_{\mathcal{P}(1\dots10\dots0)} | \mathcal{P}(\underbrace{1\dots 1}_{e} \underbrace{0\dots 0}_{N-e}) \rangle,
\end{align}
where the sum is over all permutations of $e$ excitations, i.e., in every term in superposition we have $e$ ones and $N-e$ zeros.
We assume that all the coefficients are strictly positive and we shall collectively denote them by $\alpha_\mathcal{P}$, i.e., we take $\alpha_\mathcal{P} > 0$.
Note that the highest  number of terms in the superposition is obtained for $e = (N-1)/2$ (recall that $N$ is odd) and according to the Stirling approximation it scales as $2^N/\sqrt{N}$.
We show that for all these exponentially many continuous parameters, the following even mixture
\begin{align}
	\rho = \frac12\proj{D^e_N} + \frac12\proj{D^{N-e}_N},
	\label{EQ_DICKE_NOCORR}
\end{align}
has vanishing all $N$-partite correlation functions and simultaneously it is genuinely $N$-partite entangled.

The former statement follows immediately from the results in~\cite{Us2015}.
Namely, one verifies that $| D^{N-e}_N \rangle$ is the anti-state to the generalised Dicke state $| D^{e}_N \rangle$.
Any state of odd number of qubits equally mixed with its anti-state has no $N$-partite correlation functions.
The following theorem proves genuine multipartite entanglement.
\begin{theorem}
\label{TH_Dicke}
For all $\alpha_{\mathcal{P}} > 0$ the state~\eqref{EQ_DICKE_NOCORR} is genuinely $N$-partite entangled.
\end{theorem}
\begin{proof}
We shall prove that no bi-product state exists in the subspace spanned by $\{ \ket{D^e_N}, | D^{N-e}_N \rangle \}$ if all $\alpha_{\mathcal{P}} > 0$.
The following simple observation will be utilised:
correlation functions of a bi-product state across $A:B$ partition satisfy
\begin{equation}
T_{0\dots0xx0\dots0} T_{0\dots0yy0\dots0} = T_{0\dots0xy0\dots0} T_{0\dots0yx0\dots0},
\label{EQ_BISEP}
\end{equation}
where the first non-zero index is for the last particle in $A$, and the second non-zero index is for the first particle in $B$.
We now prove that this condition is not satisfied by any pure state in the considered subspace.
Hence it contains no bi-product states and it follows that also all the mixed states with this support are genuinely $N$-partite entangled.

An arbitrary pure state in the considered subspace can be written as
\begin{align}
	\ket{\phi} = a \ket{D^e_N} + b \ket{D^{N-e}_N},
\end{align}
where $a$ and $b$ are normalised complex coefficients.
Without loss of generality we focus on the correlation functions between the last two particles:
\begin{eqnarray}
	T_{0\dots0jk} & = & |a|^2 T_{0\dots0jk}(D^e_N) + |b|^2 T_{0\dots0jk}(D^{N-e}_N) \label{EQ_Tjk} \\
	& + & a^* b\bra{D^e_N} \sigma_0 \otimes \dots \otimes \sigma_0 \otimes \sigma_j \otimes \sigma_k \ket{D^{N-e}_N} \nonumber\\
	& + & a b^*\bra{D^{N-e}_N} \sigma_0 \otimes \dots \otimes \sigma_0 \otimes \sigma_j \otimes \sigma_k \ket{D^{e}_N}, \nonumber 
\end{eqnarray}
for $j,k = x,y$. 
Note that applying $\sigma_j \otimes \sigma_k$ to the states $\ket{D^e_N}$ and $\ket{D^{N-e}_N}$ does not change their excitation parity. 
Since $N$ is odd, $\ket{D^e_N}$ and $\ket{D^{N-e}_N}$ have opposite excitation parity.
Thus the last two terms in  \eqref{EQ_Tjk} vanish.
Furthermore, the bipartite correlation functions of the anti-state to the generalised Dicke state are the same as in the original state.
We conclude that the bipartite correlations of any state $\ket{\phi}$ are the same as those of the generalised Dicke state $\ket{D^e_N}$.
One now readily verifies that for the generalised Dicke state we have:
\begin{eqnarray}
T_{0\dots 0 xy} & = & T_{0\dots 0 yx} = 0, \label{EQ_DICKE_VANISH} \\
T_{0\dots 0 xx} & = & T_{0\dots 0 yy} = \sum_{\mathcal{P}} \alpha_{\mathcal{P}(1\dots10\dots0)01} \alpha_{\mathcal{P}(1\dots10\dots0)10}, \nonumber
\end{eqnarray}
where the sum is over all permutations of $e-1$ excitations on $N-2$ positions.
Since all $\alpha_{\mathcal{P}} > 0$, Eq.~\eqref{EQ_BISEP} is never satisfied. 
The same argument holds for arbitrary partitions $A:B$.
 \end{proof}

 \subsection{Limits on entanglement witnesses based on bipartite correlations}
 
Note that the proof of genuine $N$-partite entanglement of the state in Eq.~\eqref{EQ_DICKE_NOCORR} uses solely its bipartite correlation functions. 
Furthermore, it relies on the fact that some of these correlations vanish, as in Eq.~\eqref{EQ_DICKE_VANISH}. 
Naturally one would wonder if it is possible to conclude the genuine multipartite entanglement from only non-zero bipartite correlation functions. 
This is important especially in view of entanglement witnesses which are combinations of correlation functions and therefore are insensitive to the vanishing correlation functions~\cite{PhysRevA.72.022340,JOptSocAmB.24.275,PhysRevA.86.042339,PhysRevA.87.034301,Science.344.1256,JPhysA.47.424024}.

We now show that in general the vanishing bipartite correlation functions are important for revealing of genuine $N$-partite entanglement.
Without taking them into account even entanglement of some manifestly genuinely $N$-partite entangled Dicke states is not detectable.
The Dicke state with $e$ excitations is defined by all the coefficients $\alpha_{\mathcal{P}} = 1/\sqrt{{N \choose e}}$.
It is a permutation-invariant state with the following non-vanishing bipartite correlation functions:
\begin{align}
	&T_{\mathcal{P}(xx0\dots0)} =T_{\mathcal{P}(yy0\dots0)} = \frac{2}{\binom{N}{e}}\binom{N-2}{e-1},\\
	&T_{\mathcal{P}(zz0\dots0)} = \frac{1}{\binom{N}{e}}\left\{\binom{N-2}{e}+\binom{N-2}{e-2}-2\binom{N-2}{e-1}\right\}.
\end{align}
Using the property of the binomial coefficients, $\binom{n-1}{k-1}+\binom{n-1}{k} =\binom{n}{k}$, one verifies that
\begin{align}
	T_{xx0\dots0}+T_{yy0\dots0}+T_{zz0\dots0} = 1. \label{EQ_SUMT}
\end{align}
Therefore, as long as the correlation functions in \eqref{EQ_SUMT} are non-negative, we can always construct a pure single-qubit state $\ket{\phi}$, 
with Bloch vector $(\sqrt{T_{xx0\dots0}},\sqrt{T_{yy0\dots0}},\sqrt{T_{zz0\dots0}})$, so that the tensor product $\ket{\phi}\otimes\dots\otimes\ket{\phi}$ mimics all the non-zero bipartite correlations of the Dicke state.
The non-negativity of all the terms in \eqref{EQ_SUMT} is satisfied for
\begin{align}
N \ge \left \lceil \tfrac{1}{2}(1+4e + \sqrt{1+8e}) \right \rceil,
	\label{EQ_POS}
\end{align}
where $\left \lceil x \right \rceil$ denotes the smallest integer greater or equal $x$.
For example, the non-zero bipartite correlation functions of the $\ket{W}$ state, i.e. Dicke state with $e=1$, are compatible with the correlation functions of the product state for all $N \ge 4$, hence practically for all the $\ket{W}$ states.
For such states the non-zero bipartite correlations alone are not able to reveal genuine $N$-partite entanglement. 
However, when combined with the vanishing bipartite correlations a suitable proof may be found as we illustrated above.


\section{General N}
\label{general}

We would like to present here an observation which in a simple way characterises all known facts about existence of anti-states for both $N$ even and odd.
It provides yet another piece of evidence that arbitrary genuinely $N$-partite entangled pure state of even number of qubits does not admit an anti-state.

Consider a state $\ket{\psi}$ endowed with correlation tensor $T_{j_1 \dots j_N}$. 
Recall that its anti-state is defined by having correlation tensor elements given by $ - T_{j_1 \dots j_N}$, for all indices $j_n = x,y,z$.
One way of obtaining an anti-state would be to apply on an odd number of qubits a local operation which maps
\begin{equation}
\vec x \to - \vec x, \qquad \vec y \to - \vec y, \qquad \vec z \to - \vec z.
\label{UNOT}
\end{equation}
However, it is well-known that such a local operation, called universal-not gate~\cite{JModOpt.47.211}, is not present within quantum formalism as it is anti-unitary.
On the level of multiple qubits one can to some degree overcome this restriction.
Namely, note that mathematically one obtains \eqref{UNOT} by applying $\sigma_y$ operation and partial transposition.
The effect of $\sigma_y$ is to invert $\vec x \to - \vec x$ and $\vec z \to - \vec z$, and the effect of partial transposition is to flip the remaining axis $\vec y \to - \vec y$.
If partial transposition is applied on a subsystem $A$ of a pure state entangled across $A:B$ it results in a matrix with negative eigenvalues~\cite{PhysRevLett.77.1413,PhysLettA.223.1}.
Hence, this method leads to a physically meaningful anti-state only for original states with odd total number of qubits (as applying partial transposition on every individual qubit results in a transposition, which is a completely positive map)
or having a subsystem $A$ with an odd number of qubits in a product state.
For example, by applying $\sigma_y$ and partial transposition on every single qubit or by taking $A$ as the first qubit this procedure will produce an anti-state to three-qubit $\ket{0}\ket{\psi^-}$ state, 
but no anti-state to four-qubit $\ket{\psi^-} \ket{\psi^-}$ state, and indeed any genuinely multipartite entangled pure state of even number of qubits.
Of course global operations may exist that produce anti-states in a completely different way, but nevertheless it is appealing that this simple procedure recovers all that is known about anti-states at present date.

\subsection{Impossibility of inverting all correlation functions between even number of observers}

We would like to finish with one more observation contrasting even and odd lower order correlation functions in an anti-state for general $N$.
The anti-states constructed in Ref.~\cite{Us2015} have opposite correlation functions between arbitrary odd number of observers, as compared to the original state.
The correlation functions between an arbitrary even number of observers are the same as in the original state.
In contrast, there is no state in which all the correlation functions between arbitrary even number of observers are opposite.

\begin{theorem}
Any pure state $\ket{\psi}$ of $N$ qubits does not admit state $\ket{\psi'}$ in which all the $k$-partite correlation functions, for all even $k$, are opposite.
\end{theorem}
\begin{proof}
By contradiction.
Let us build an anti-state to the hypothetical state $\ket{\psi'}$ according to prescription of Ref.~\cite{Us2015}.
Denote it $\ket{\bar{\psi'}}$ and note that it has opposite all the correlation functions between even and odd number of observers, as compared to the original state $\ket{\psi}$.
Therefore the even mixture 
\begin{equation}
\rho=\frac 1 2 \proj{\psi}+\frac 1 2 \proj{\bar{\psi'}},
\end{equation}
has no correlations whatsoever, including expectation values of local observables, i.e. $\rho$ is a white noise $\mathbb{I}/2^N$.
However, this is not possible since the rank of $\rho$ is $2$, while the white noise must have rank $2^N$.
\end{proof}

\section{Conclusions}

We provided non-trivial examples of genuinely multiparty entangled states of even number $N$ of qubits that simultaneously have vanishing $N$-partite correlation functions.
We showed that they violate suitable Bell-type inequalities.
The states have rank $4$ and rank $3$, respectively, and we gave a compelling evidence supporting the conjecture that rank $2$ examples do not exist.
This is in contrast to multipartite systems with odd number of qubits and explains why only such cases were considered up to date.
We also extended previously known examples using techniques that among others show limits to entanglement detection with bipartite correlation functions only.

The states discussed here opened a debate on rigorous quantification of genuine multipartite classical and quantum correlations that led to the formulation of natural postulates such quantifiers should satisfy~\cite{PhysRevA.83.012312}.
We hope that the examples provided here will be a useful testbed for candidate identifiers and will help to find computable measures that will enable a deeper analysis of multipartite experiments.

\section{Acknowledgments}

This work is supported by the Singapore Ministry of Education Academic Research Fund Tier 2 project MOE2015-T2-2-034 and NCN Grant No. 2014/14/M/ST2/00818.
We acknowledge the support of this work by the EU (ERC QOLAPS). 
L.K. thanks the international PhD programme ExQM from the Elite Network of Bavaria for support.
M.C.T acknowledges support from the NSF-funded Physics Frontier Center at the JQI. 
\section{Appendix}

\begin{theorem}
All $N$-partite correlation functions of the state in Eq. (\ref{EVEN_NO_CORR}) vanish.
\end{theorem}
\begin{proof}
To simplify notation we divide all $N$ observers into Alice and Bob, each in possession of $N/2$ qubits.
The $N$-partite quantum correlation functions are written as $T_{AB}(\rho)$, with $A$ and $B$ being sequences, each of length $N/2$, of indices $x,y,z$.
For example,
\begin{align}
	2 \, T_{AB}(\psi_2) =& \bra{\psi_2}\sigma_A\otimes\sigma_B\ket{\psi_2}\\
	=&\bra{1\dots1}\sigma_A\ket{1\dots1}\bra{\bar \psi}\sigma_B\ket{\bar \psi}\nonumber\\
	+ & \bra{\bar \psi}\sigma_A\ket{\bar \psi}\bra{1\dots1}\sigma_B\ket{1\dots1}\nonumber\\
	- & \bra{1\dots1}\sigma_A\ket{\bar \psi}\bra{\bar \psi}\sigma_B\ket{1\dots1}\nonumber\\
	- & \bra{\bar \psi}\sigma_A\ket{1\dots1}\bra{1\dots1}\sigma_B\ket{\bar \psi}.
\end{align}
Note that $\ket{\bar \psi} = \sigma_X \otimes \sigma_X \ket{\psi}$, where each $\sigma_X \equiv\sigma_x\otimes\dots\otimes\sigma_x$ operates on all the qubits of Alice/Bob.
Furthermore, for Alice we have $\sigma_X \sigma_A \sigma_X = (-1)^{a}\sigma_A$, where $a$ is the number of $x$ indices appearing in the sequence $A$.
Similarly, $\sigma_X\sigma_B\sigma_X = (-1)^{b}\sigma_B$, where $b$ is the number of $x$ indices appearing in the sequence $B$.
Therefore, if $a+b$ is even the $N$-partite correlation functions of $\rho$ read:
\begin{align}
	T_{AB}(\rho) =&\frac{1}{4}\sum_{i=1}^4 T_{AB}(\psi_i) \\
	=&\bra{00\dots0}\sigma_A\ket{00\dots0}\bra{\psi}\sigma_B\ket{\psi}\nonumber\\
	+&\bra{\psi}\sigma_A\ket{\psi}\bra{00\dots0}\sigma_B\ket{00\dots0}\nonumber\\
	+&\bra{\phi}\sigma_A\ket{\phi}\bra{\psi}\sigma_B\ket{\psi}\nonumber\\
	+&\bra{\psi}\sigma_A\ket{\psi}\bra{\phi}\sigma_B\ket{\phi}.
\end{align}
But due to an odd number of excitations in $\ket{\phi}$ we have that $\bra{00\dots0}\sigma_A\ket{00\dots0}$ and $\bra{\phi}\sigma_A\ket{\phi}$ are either both zero or have opposite sign (and the same for Bob).
We thus arrive at vanishing $N$-partite correlation functions of $\rho$.

If $a+b$ is odd, we instead have
\begin{align}
	T_{AB}(\psi_1)+T_{AB}(\psi_2) = &\bra{0\dots0}\sigma_A\ket{\psi}\bra{\psi}\sigma_B\ket{0\dots0}\nonumber\\
	+ & \bra{\psi}\sigma_A\ket{0\dots0}\bra{0\dots0}\sigma_B\ket{\psi}. \label{EQ_SUM2ODD}
\end{align}
Since by our assumption $\ket{\psi}$ is a superposition of states with only one excitation, 
both terms above vanish unless $A$ and $B$ each has only one $x$ or $y$ index.
Hence together they must have in total an even number of $x$ and $y$ indices.
But the number of $x$ indices, i.e. $a+b$ is assumed to be odd, so the number of $y$ indices must also be odd. 
Therefore both terms in \eqref{EQ_SUM2ODD} are imaginary and since they are complex adjoints of each other the sum $T_{AB}(\psi_1)+T_{AB}(\psi_2)$ vanishes.
Meanwhile, the contribution from $\ket{\psi_3}$ is
\begin{align}
	T_{AB}(\psi_3) = \bra{\phi}\sigma_A\ket{\phi}\bra{\psi}\sigma_B\ket{\psi}.
\end{align}
For this to be non-zero, $\sigma_A$ must be $\sigma_z\otimes\dots\otimes\sigma_z$, and therefore has no $x$ in the sequence: $a=0$. 
Since $a+b$ is odd, $b$ must be odd. 
But $\bra{\psi}\sigma_B\ket{\psi}$ is non-zero only if $B$ contains an even number of $x$ and $y$ indices in total.
Thus the number of $y$ indices must be odd leading to an imaginary $T_{AB}(\psi_3)$.
Since the correlation function is defined as the average of real numbers it is always read valued. 
We conclude that $T_{AB}(\psi_3) = 0$.
The same argument applies to $\ket{\psi_4}$.
\end{proof}

\bibliographystyle{apsrev4-1}
\bibliography{ent_no_corr}

\begin{thebibliography}{27}%
\makeatletter
\providecommand \@ifxundefined [1]{%
 \@ifx{#1\undefined}
}%
\providecommand \@ifnum [1]{%
 \ifnum #1\expandafter \@firstoftwo
 \else \expandafter \@secondoftwo
 \fi
}%
\providecommand \@ifx [1]{%
 \ifx #1\expandafter \@firstoftwo
 \else \expandafter \@secondoftwo
 \fi
}%
\providecommand \natexlab [1]{#1}%
\providecommand \enquote  [1]{``#1''}%
\providecommand \bibnamefont  [1]{#1}%
\providecommand \bibfnamefont [1]{#1}%
\providecommand \citenamefont [1]{#1}%
\providecommand \href@noop [0]{\@secondoftwo}%
\providecommand \href [0]{\begingroup \@sanitize@url \@href}%
\providecommand \@href[1]{\@@startlink{#1}\@@href}%
\providecommand \@@href[1]{\endgroup#1\@@endlink}%
\providecommand \@sanitize@url [0]{\catcode `\\12\catcode `\$12\catcode
  `\&12\catcode `\#12\catcode `\^12\catcode `\_12\catcode `\%12\relax}%
\providecommand \@@startlink[1]{}%
\providecommand \@@endlink[0]{}%
\providecommand \url  [0]{\begingroup\@sanitize@url \@url }%
\providecommand \@url [1]{\endgroup\@href {#1}{\urlprefix }}%
\providecommand \urlprefix  [0]{URL }%
\providecommand \Eprint [0]{\href }%
\providecommand \doibase [0]{http://dx.doi.org/}%
\providecommand \selectlanguage [0]{\@gobble}%
\providecommand \bibinfo  [0]{\@secondoftwo}%
\providecommand \bibfield  [0]{\@secondoftwo}%
\providecommand \translation [1]{[#1]}%
\providecommand \BibitemOpen [0]{}%
\providecommand \bibitemStop [0]{}%
\providecommand \bibitemNoStop [0]{.\EOS\space}%
\providecommand \EOS [0]{\spacefactor3000\relax}%
\providecommand \BibitemShut  [1]{\csname bibitem#1\endcsname}%
\let\auto@bib@innerbib\@empty
\bibitem [{\citenamefont {Bennett}\ \emph {et~al.}(1996)\citenamefont
  {Bennett}, \citenamefont {DiVincenzo}, \citenamefont {Smolin},\ and\
  \citenamefont {Wootters}}]{locc}%
  \BibitemOpen
  \bibfield  {author} {\bibinfo {author} {\bibfnamefont {C.~H.}\ \bibnamefont
  {Bennett}}, \bibinfo {author} {\bibfnamefont {D.~P.}\ \bibnamefont
  {DiVincenzo}}, \bibinfo {author} {\bibfnamefont {J.~A.}\ \bibnamefont
  {Smolin}}, \ and\ \bibinfo {author} {\bibfnamefont {W.~K.}\ \bibnamefont
  {Wootters}},\ }\href@noop {} {\bibfield  {journal} {\bibinfo  {journal}
  {Phys. Rev. A}\ }\textbf {\bibinfo {volume} {54}},\ \bibinfo {pages} {3824}
  (\bibinfo {year} {1996})}\BibitemShut {NoStop}%
\bibitem [{\citenamefont {Horodecki}\ \emph {et~al.}(2009)\citenamefont
  {Horodecki}, \citenamefont {Horodecki}, \citenamefont {Horodecki},\ and\
  \citenamefont {Horodecki}}]{RevModPhys.81.865}%
  \BibitemOpen
  \bibfield  {author} {\bibinfo {author} {\bibfnamefont {R.}~\bibnamefont
  {Horodecki}}, \bibinfo {author} {\bibfnamefont {P.}~\bibnamefont
  {Horodecki}}, \bibinfo {author} {\bibfnamefont {M.}~\bibnamefont
  {Horodecki}}, \ and\ \bibinfo {author} {\bibfnamefont {K.}~\bibnamefont
  {Horodecki}},\ }\href@noop {} {\bibfield  {journal} {\bibinfo  {journal}
  {Rev. Mod. Phys.}\ }\textbf {\bibinfo {volume} {81}},\ \bibinfo {pages} {865}
  (\bibinfo {year} {2009})}\BibitemShut {NoStop}%
\bibitem [{\citenamefont {Hassan}\ and\ \citenamefont
  {Joag}(2007)}]{QuantInfComp.8.773}%
  \BibitemOpen
  \bibfield  {author} {\bibinfo {author} {\bibfnamefont {A.~S.~M.}\
  \bibnamefont {Hassan}}\ and\ \bibinfo {author} {\bibfnamefont {P.~S.}\
  \bibnamefont {Joag}},\ }\href@noop {} {\bibfield  {journal} {\bibinfo
  {journal} {Quant. Inf. Comp.}\ }\textbf {\bibinfo {volume} {8}},\ \bibinfo
  {pages} {773} (\bibinfo {year} {2007})}\BibitemShut {NoStop}%
\bibitem [{\citenamefont {Hassan}\ and\ \citenamefont
  {Joag}(2008)}]{PhysRevA.77.062334}%
  \BibitemOpen
  \bibfield  {author} {\bibinfo {author} {\bibfnamefont {A.~S.~M.}\
  \bibnamefont {Hassan}}\ and\ \bibinfo {author} {\bibfnamefont {P.~S.}\
  \bibnamefont {Joag}},\ }\href@noop {} {\bibfield  {journal} {\bibinfo
  {journal} {Phys. Rev. A}\ }\textbf {\bibinfo {volume} {77}},\ \bibinfo
  {pages} {062334} (\bibinfo {year} {2008})}\BibitemShut {NoStop}%
\bibitem [{\citenamefont {Hassan}\ and\ \citenamefont {Joag}(2009)}]{I}%
  \BibitemOpen
  \bibfield  {author} {\bibinfo {author} {\bibfnamefont {A.~S.~M.}\
  \bibnamefont {Hassan}}\ and\ \bibinfo {author} {\bibfnamefont {P.~S.}\
  \bibnamefont {Joag}},\ }\href@noop {} {\bibfield  {journal} {\bibinfo
  {journal} {Phys. Rev. A}\ }\textbf {\bibinfo {volume} {80}},\ \bibinfo
  {pages} {042302} (\bibinfo {year} {2009})}\BibitemShut {NoStop}%
\bibitem [{\citenamefont {Tran}\ \emph {et~al.}(2015)\citenamefont {Tran},
  \citenamefont {Daki\'c}, \citenamefont {Arnault}, \citenamefont {Laskowski},\
  and\ \citenamefont {Paterek}}]{PhysRevA.92.050301}%
  \BibitemOpen
  \bibfield  {author} {\bibinfo {author} {\bibfnamefont {M.~C.}\ \bibnamefont
  {Tran}}, \bibinfo {author} {\bibfnamefont {B.}~\bibnamefont {Daki\'c}},
  \bibinfo {author} {\bibfnamefont {F.}~\bibnamefont {Arnault}}, \bibinfo
  {author} {\bibfnamefont {W.}~\bibnamefont {Laskowski}}, \ and\ \bibinfo
  {author} {\bibfnamefont {T.}~\bibnamefont {Paterek}},\ }\href@noop {}
  {\bibfield  {journal} {\bibinfo  {journal} {Phys. Rev. A}\ }\textbf {\bibinfo
  {volume} {92}},\ \bibinfo {pages} {050301(R)} (\bibinfo {year}
  {2015})}\BibitemShut {NoStop}%
\bibitem [{\citenamefont {Tran}\ \emph {et~al.}(2016)\citenamefont {Tran},
  \citenamefont {Daki\'c}, \citenamefont {Laskowski},\ and\ \citenamefont
  {Paterek}}]{Random2016}%
  \BibitemOpen
  \bibfield  {author} {\bibinfo {author} {\bibfnamefont {M.~C.}\ \bibnamefont
  {Tran}}, \bibinfo {author} {\bibfnamefont {B.}~\bibnamefont {Daki\'c}},
  \bibinfo {author} {\bibfnamefont {W.}~\bibnamefont {Laskowski}}, \ and\
  \bibinfo {author} {\bibfnamefont {T.}~\bibnamefont {Paterek}},\ }\href@noop
  {} {\bibfield  {journal} {\bibinfo  {journal} {Phys. Rev. A}\ }\textbf
  {\bibinfo {volume} {94}},\ \bibinfo {pages} {042302} (\bibinfo {year}
  {2016})}\BibitemShut {NoStop}%
\bibitem [{\citenamefont {Kaszlikowski}\ \emph {et~al.}(2008)\citenamefont
  {Kaszlikowski}, \citenamefont {{Sen~De}}, \citenamefont {Sen}, \citenamefont
  {Vedral},\ and\ \citenamefont {Winter}}]{PhysRevLett.101.070502}%
  \BibitemOpen
  \bibfield  {author} {\bibinfo {author} {\bibfnamefont {D.}~\bibnamefont
  {Kaszlikowski}}, \bibinfo {author} {\bibfnamefont {A.}~\bibnamefont
  {{Sen~De}}}, \bibinfo {author} {\bibfnamefont {U.}~\bibnamefont {Sen}},
  \bibinfo {author} {\bibfnamefont {V.}~\bibnamefont {Vedral}}, \ and\ \bibinfo
  {author} {\bibfnamefont {A.}~\bibnamefont {Winter}},\ }\href@noop {}
  {\bibfield  {journal} {\bibinfo  {journal} {Phys. Rev. Lett.}\ }\textbf
  {\bibinfo {volume} {101}},\ \bibinfo {pages} {070502} (\bibinfo {year}
  {2008})}\BibitemShut {NoStop}%
\bibitem [{\citenamefont {Laskowski}\ \emph {et~al.}(2012)\citenamefont
  {Laskowski}, \citenamefont {Markiewicz}, \citenamefont {Paterek},\ and\
  \citenamefont {Wie\'sniak}}]{PhysRevA.86.032105}%
  \BibitemOpen
  \bibfield  {author} {\bibinfo {author} {\bibfnamefont {W.}~\bibnamefont
  {Laskowski}}, \bibinfo {author} {\bibfnamefont {M.}~\bibnamefont
  {Markiewicz}}, \bibinfo {author} {\bibfnamefont {T.}~\bibnamefont {Paterek}},
  \ and\ \bibinfo {author} {\bibfnamefont {M.}~\bibnamefont {Wie\'sniak}},\
  }\href@noop {} {\bibfield  {journal} {\bibinfo  {journal} {Phys. Rev. A}\
  }\textbf {\bibinfo {volume} {86}},\ \bibinfo {pages} {032105} (\bibinfo
  {year} {2012})}\BibitemShut {NoStop}%
\bibitem [{\citenamefont {Schwemmer}\ \emph {et~al.}(2015)\citenamefont
  {Schwemmer}, \citenamefont {Knips}, \citenamefont {Tran}, \citenamefont
  {de~Rosier}, \citenamefont {Laskowski}, \citenamefont {Paterek},\ and\
  \citenamefont {Weinfurter}}]{Us2015}%
  \BibitemOpen
  \bibfield  {author} {\bibinfo {author} {\bibfnamefont {C.}~\bibnamefont
  {Schwemmer}}, \bibinfo {author} {\bibfnamefont {L.}~\bibnamefont {Knips}},
  \bibinfo {author} {\bibfnamefont {M.~C.}\ \bibnamefont {Tran}}, \bibinfo
  {author} {\bibfnamefont {A.}~\bibnamefont {de~Rosier}}, \bibinfo {author}
  {\bibfnamefont {W.}~\bibnamefont {Laskowski}}, \bibinfo {author}
  {\bibfnamefont {T.}~\bibnamefont {Paterek}}, \ and\ \bibinfo {author}
  {\bibfnamefont {H.}~\bibnamefont {Weinfurter}},\ }\href@noop {} {\bibfield
  {journal} {\bibinfo  {journal} {Phys Rev. Lett.}\ }\textbf {\bibinfo {volume}
  {114}},\ \bibinfo {pages} {180501} (\bibinfo {year} {2015})}\BibitemShut
  {NoStop}%
\bibitem [{\citenamefont {T\'oth}\ and\ \citenamefont
  {G\"uhne}(2005)}]{PhysRevA.72.022340}%
  \BibitemOpen
  \bibfield  {author} {\bibinfo {author} {\bibfnamefont {G.}~\bibnamefont
  {T\'oth}}\ and\ \bibinfo {author} {\bibfnamefont {O.}~\bibnamefont
  {G\"uhne}},\ }\href@noop {} {\bibfield  {journal} {\bibinfo  {journal} {Phys.
  Rev. A}\ }\textbf {\bibinfo {volume} {72}},\ \bibinfo {pages} {022340}
  (\bibinfo {year} {2005})}\BibitemShut {NoStop}%
\bibitem [{\citenamefont {T\'oth}(2007)}]{JOptSocAmB.24.275}%
  \BibitemOpen
  \bibfield  {author} {\bibinfo {author} {\bibfnamefont {G.}~\bibnamefont
  {T\'oth}},\ }\href@noop {} {\bibfield  {journal} {\bibinfo  {journal} {J.
  Opt. Soc. Am. B}\ }\textbf {\bibinfo {volume} {24}},\ \bibinfo {pages} {275}
  (\bibinfo {year} {2007})}\BibitemShut {NoStop}%
\bibitem [{\citenamefont {Wie\'sniak}\ \emph {et~al.}(2012)\citenamefont
  {Wie\'sniak}, \citenamefont {Nawareg},\ and\ \citenamefont
  {\.Zukowski}}]{PhysRevA.86.042339}%
  \BibitemOpen
  \bibfield  {author} {\bibinfo {author} {\bibfnamefont {M.}~\bibnamefont
  {Wie\'sniak}}, \bibinfo {author} {\bibfnamefont {M.}~\bibnamefont {Nawareg}},
  \ and\ \bibinfo {author} {\bibfnamefont {M.}~\bibnamefont {\.Zukowski}},\
  }\href@noop {} {\bibfield  {journal} {\bibinfo  {journal} {Phys. Rev. A}\
  }\textbf {\bibinfo {volume} {86}},\ \bibinfo {pages} {042339} (\bibinfo
  {year} {2012})}\BibitemShut {NoStop}%
\bibitem [{\citenamefont {Markiewicz}\ \emph {et~al.}(2013)\citenamefont
  {Markiewicz}, \citenamefont {Laskowski}, \citenamefont {Paterek},\ and\
  \citenamefont {\.Zukowski}}]{PhysRevA.87.034301}%
  \BibitemOpen
  \bibfield  {author} {\bibinfo {author} {\bibfnamefont {M.}~\bibnamefont
  {Markiewicz}}, \bibinfo {author} {\bibfnamefont {W.}~\bibnamefont
  {Laskowski}}, \bibinfo {author} {\bibfnamefont {T.}~\bibnamefont {Paterek}},
  \ and\ \bibinfo {author} {\bibfnamefont {M.}~\bibnamefont {\.Zukowski}},\
  }\href@noop {} {\bibfield  {journal} {\bibinfo  {journal} {Phys. Rev. A}\
  }\textbf {\bibinfo {volume} {87}},\ \bibinfo {pages} {034301} (\bibinfo
  {year} {2013})}\BibitemShut {NoStop}%
\bibitem [{\citenamefont {Tura}\ \emph
  {et~al.}(2014{\natexlab{a}})\citenamefont {Tura}, \citenamefont {Augusiak},
  \citenamefont {Sainz}, \citenamefont {Vertesi}, \citenamefont {Lewenstein},\
  and\ \citenamefont {Acin}}]{Science.344.1256}%
  \BibitemOpen
  \bibfield  {author} {\bibinfo {author} {\bibfnamefont {J.}~\bibnamefont
  {Tura}}, \bibinfo {author} {\bibfnamefont {R.}~\bibnamefont {Augusiak}},
  \bibinfo {author} {\bibfnamefont {A.~B.}\ \bibnamefont {Sainz}}, \bibinfo
  {author} {\bibfnamefont {T.}~\bibnamefont {Vertesi}}, \bibinfo {author}
  {\bibfnamefont {M.}~\bibnamefont {Lewenstein}}, \ and\ \bibinfo {author}
  {\bibfnamefont {A.}~\bibnamefont {Acin}},\ }\href@noop {} {\bibfield
  {journal} {\bibinfo  {journal} {Science}\ }\textbf {\bibinfo {volume}
  {344}},\ \bibinfo {pages} {1256} (\bibinfo {year}
  {2014}{\natexlab{a}})}\BibitemShut {NoStop}%
\bibitem [{\citenamefont {Tura}\ \emph
  {et~al.}(2014{\natexlab{b}})\citenamefont {Tura}, \citenamefont {Sainz},
  \citenamefont {Vertesi}, \citenamefont {Acin}, \citenamefont {Lewenstein},\
  and\ \citenamefont {Augusiak}}]{JPhysA.47.424024}%
  \BibitemOpen
  \bibfield  {author} {\bibinfo {author} {\bibfnamefont {J.}~\bibnamefont
  {Tura}}, \bibinfo {author} {\bibfnamefont {A.~B.}\ \bibnamefont {Sainz}},
  \bibinfo {author} {\bibfnamefont {T.}~\bibnamefont {Vertesi}}, \bibinfo
  {author} {\bibfnamefont {A.}~\bibnamefont {Acin}}, \bibinfo {author}
  {\bibfnamefont {M.}~\bibnamefont {Lewenstein}}, \ and\ \bibinfo {author}
  {\bibfnamefont {R.}~\bibnamefont {Augusiak}},\ }\href@noop {} {\bibfield
  {journal} {\bibinfo  {journal} {J. Phys. A}\ }\textbf {\bibinfo {volume}
  {47}},\ \bibinfo {pages} {424024} (\bibinfo {year}
  {2014}{\natexlab{b}})}\BibitemShut {NoStop}%
\bibitem [{\citenamefont {Kirkpatrick}\ \emph {et~al.}(1983)\citenamefont
  {Kirkpatrick}, \citenamefont {Gelatt},\ and\ \citenamefont
  {Vecchi}}]{Kirkpatrick}%
  \BibitemOpen
  \bibfield  {author} {\bibinfo {author} {\bibfnamefont {S.}~\bibnamefont
  {Kirkpatrick}}, \bibinfo {author} {\bibfnamefont {C.~D.}\ \bibnamefont
  {Gelatt}}, \ and\ \bibinfo {author} {\bibfnamefont {M.~P.}\ \bibnamefont
  {Vecchi}},\ }\href@noop {} {\bibfield  {journal} {\bibinfo  {journal}
  {Science}\ }\textbf {\bibinfo {volume} {220}},\ \bibinfo {pages} {671}
  (\bibinfo {year} {1983})}\BibitemShut {NoStop}%
\bibitem [{\citenamefont {Horodecki}\ \emph {et~al.}(2003)\citenamefont
  {Horodecki}, \citenamefont {Smolin}, \citenamefont {Terhal},\ and\
  \citenamefont {Thapliyal}}]{TheorComputSci.292.589}%
  \BibitemOpen
  \bibfield  {author} {\bibinfo {author} {\bibfnamefont {P.}~\bibnamefont
  {Horodecki}}, \bibinfo {author} {\bibfnamefont {J.~A.}\ \bibnamefont
  {Smolin}}, \bibinfo {author} {\bibfnamefont {B.~M.}\ \bibnamefont {Terhal}},
  \ and\ \bibinfo {author} {\bibfnamefont {A.~V.}\ \bibnamefont {Thapliyal}},\
  }\href@noop {} {\bibfield  {journal} {\bibinfo  {journal} {Theor. Comput.
  Sci.}\ }\textbf {\bibinfo {volume} {292}},\ \bibinfo {pages} {589} (\bibinfo
  {year} {2003})}\BibitemShut {NoStop}%
\bibitem [{\citenamefont {Tran}\ \emph {et~al.}(2014)\citenamefont {Tran},
  \citenamefont {Laskowski},\ and\ \citenamefont {Paterek}}]{JPhysA.47.424025}%
  \BibitemOpen
  \bibfield  {author} {\bibinfo {author} {\bibfnamefont {M.~C.}\ \bibnamefont
  {Tran}}, \bibinfo {author} {\bibfnamefont {W.}~\bibnamefont {Laskowski}}, \
  and\ \bibinfo {author} {\bibfnamefont {T.}~\bibnamefont {Paterek}},\
  }\href@noop {} {\bibfield  {journal} {\bibinfo  {journal} {J. Phys. A}\
  }\textbf {\bibinfo {volume} {47}},\ \bibinfo {pages} {424025} (\bibinfo
  {year} {2014})}\BibitemShut {NoStop}%
\bibitem [{\citenamefont {Peres}(1996)}]{PhysRevLett.77.1413}%
  \BibitemOpen
  \bibfield  {author} {\bibinfo {author} {\bibfnamefont {A.}~\bibnamefont
  {Peres}},\ }\href@noop {} {\bibfield  {journal} {\bibinfo  {journal} {Phys.
  Rev. Lett.}\ }\textbf {\bibinfo {volume} {77}},\ \bibinfo {pages} {1413}
  (\bibinfo {year} {1996})}\BibitemShut {NoStop}%
\bibitem [{\citenamefont {Horodecki}\ \emph {et~al.}(1996)\citenamefont
  {Horodecki}, \citenamefont {Horodecki},\ and\ \citenamefont
  {Horodecki}}]{PhysLettA.223.1}%
  \BibitemOpen
  \bibfield  {author} {\bibinfo {author} {\bibfnamefont {M.}~\bibnamefont
  {Horodecki}}, \bibinfo {author} {\bibfnamefont {P.}~\bibnamefont
  {Horodecki}}, \ and\ \bibinfo {author} {\bibfnamefont {R.}~\bibnamefont
  {Horodecki}},\ }\href@noop {} {\bibfield  {journal} {\bibinfo  {journal}
  {Phys. Lett. A}\ }\textbf {\bibinfo {volume} {223}},\ \bibinfo {pages} {1}
  (\bibinfo {year} {1996})}\BibitemShut {NoStop}%
\bibitem [{\citenamefont {Laskowski}\ \emph {et~al.}(2015)\citenamefont
  {Laskowski}, \citenamefont {Vertesi},\ and\ \citenamefont
  {Wie\'sniak}}]{JPhysA.48.465301}%
  \BibitemOpen
  \bibfield  {author} {\bibinfo {author} {\bibfnamefont {W.}~\bibnamefont
  {Laskowski}}, \bibinfo {author} {\bibfnamefont {T.}~\bibnamefont {Vertesi}},
  \ and\ \bibinfo {author} {\bibfnamefont {M.}~\bibnamefont {Wie\'sniak}},\
  }\href@noop {} {\bibfield  {journal} {\bibinfo  {journal} {J. Phys. A}\
  }\textbf {\bibinfo {volume} {48}},\ \bibinfo {pages} {465301} (\bibinfo
  {year} {2015})}\BibitemShut {NoStop}%
\bibitem [{\citenamefont {Clauser}\ and\ \citenamefont {Horne}(1974)}]{CH}%
  \BibitemOpen
  \bibfield  {author} {\bibinfo {author} {\bibfnamefont {J.~F.}\ \bibnamefont
  {Clauser}}\ and\ \bibinfo {author} {\bibfnamefont {M.~A.}\ \bibnamefont
  {Horne}},\ }\href@noop {} {\bibfield  {journal} {\bibinfo  {journal} {Phys.
  Rev. D}\ }\textbf {\bibinfo {volume} {10}},\ \bibinfo {pages} {526} (\bibinfo
  {year} {1974})}\BibitemShut {NoStop}%
\bibitem [{\citenamefont {Gruca}\ \emph {et~al.}(2010)\citenamefont {Gruca},
  \citenamefont {Laskowski}, \citenamefont {\.Zukowski}, \citenamefont
  {Kiesel}, \citenamefont {Wieczorek}, \citenamefont {Schmid},\ and\
  \citenamefont {Weinfurter}}]{PhysRevA.82.012118}%
  \BibitemOpen
  \bibfield  {author} {\bibinfo {author} {\bibfnamefont {J.}~\bibnamefont
  {Gruca}}, \bibinfo {author} {\bibfnamefont {W.}~\bibnamefont {Laskowski}},
  \bibinfo {author} {\bibfnamefont {M.}~\bibnamefont {\.Zukowski}}, \bibinfo
  {author} {\bibfnamefont {N.}~\bibnamefont {Kiesel}}, \bibinfo {author}
  {\bibfnamefont {W.}~\bibnamefont {Wieczorek}}, \bibinfo {author}
  {\bibfnamefont {C.}~\bibnamefont {Schmid}}, \ and\ \bibinfo {author}
  {\bibfnamefont {H.}~\bibnamefont {Weinfurter}},\ }\href@noop {} {\bibfield
  {journal} {\bibinfo  {journal} {Phys. Rev. A}\ }\textbf {\bibinfo {volume}
  {82}},\ \bibinfo {pages} {012118} (\bibinfo {year} {2010})}\BibitemShut
  {NoStop}%
\bibitem [{\citenamefont {Brunner}\ \emph {et~al.}(2014)\citenamefont
  {Brunner}, \citenamefont {Cavalcanti}, \citenamefont {Pironio}, \citenamefont
  {Scarani},\ and\ \citenamefont {Wehner}}]{RevModPhys.86.419}%
  \BibitemOpen
  \bibfield  {author} {\bibinfo {author} {\bibfnamefont {N.}~\bibnamefont
  {Brunner}}, \bibinfo {author} {\bibfnamefont {D.}~\bibnamefont {Cavalcanti}},
  \bibinfo {author} {\bibfnamefont {S.}~\bibnamefont {Pironio}}, \bibinfo
  {author} {\bibfnamefont {V.}~\bibnamefont {Scarani}}, \ and\ \bibinfo
  {author} {\bibfnamefont {S.}~\bibnamefont {Wehner}},\ }\href@noop {}
  {\bibfield  {journal} {\bibinfo  {journal} {Rev. Mod. Phys.}\ }\textbf
  {\bibinfo {volume} {86}},\ \bibinfo {pages} {419} (\bibinfo {year}
  {2014})}\BibitemShut {NoStop}%
\bibitem [{\citenamefont {Bu{\v z}ek}\ \emph {et~al.}(2000)\citenamefont {Bu{\v
  z}ek}, \citenamefont {Hillery},\ and\ \citenamefont
  {Werner}}]{JModOpt.47.211}%
  \BibitemOpen
  \bibfield  {author} {\bibinfo {author} {\bibfnamefont {V.}~\bibnamefont
  {Bu{\v z}ek}}, \bibinfo {author} {\bibfnamefont {M.}~\bibnamefont {Hillery}},
  \ and\ \bibinfo {author} {\bibfnamefont {R.~F.}\ \bibnamefont {Werner}},\
  }\href@noop {} {\bibfield  {journal} {\bibinfo  {journal} {J. Mod. Opt.}\
  }\textbf {\bibinfo {volume} {47}},\ \bibinfo {pages} {211} (\bibinfo {year}
  {2000})}\BibitemShut {NoStop}%
\bibitem [{\citenamefont {Bennett}\ \emph {et~al.}(2011)\citenamefont
  {Bennett}, \citenamefont {Grudka}, \citenamefont {Horodecki}, \citenamefont
  {Horodecki},\ and\ \citenamefont {Horodecki}}]{PhysRevA.83.012312}%
  \BibitemOpen
  \bibfield  {author} {\bibinfo {author} {\bibfnamefont {C.~H.}\ \bibnamefont
  {Bennett}}, \bibinfo {author} {\bibfnamefont {A.}~\bibnamefont {Grudka}},
  \bibinfo {author} {\bibfnamefont {M.}~\bibnamefont {Horodecki}}, \bibinfo
  {author} {\bibfnamefont {P.}~\bibnamefont {Horodecki}}, \ and\ \bibinfo
  {author} {\bibfnamefont {R.}~\bibnamefont {Horodecki}},\ }\href@noop {}
  {\bibfield  {journal} {\bibinfo  {journal} {Phys. Rev. A}\ }\textbf {\bibinfo
  {volume} {83}},\ \bibinfo {pages} {012312} (\bibinfo {year}
  {2011})}\BibitemShut {NoStop}%
\end{thebibliography}%

\end{document}